\documentclass{sig-alternate}

\usepackage{microtype}
\usepackage{amsmath}

\usepackage{algorithm}
\usepackage{algorithmic}
\usepackage{bm}

\usepackage{epsfig}
\usepackage{epstopdf}
\usepackage{enumitem}

\setlist{itemsep=0em}
\newcommand{\Osymbol}{{\mathcal O}}
\newcommand{\BO}[1]{\Osymbol\left(#1\right)}

\newcommand{\E}[1]{\textrm{\bf E}\left[#1\right]}
\renewcommand{\Pr}[1]{\textrm{\bf Pr}\left[#1\right]}

\newcommand{\full}[1]{{#1}}
\newcommand{\short}[1]{}

\newtheorem{lemma}{Lemma}
\newtheorem{theorem}{Theorem}
\newtheorem{definition}{Definition}
\newtheorem{corollary}{Corollary}
\newtheorem{exmp}{Example}[section]

\begin{document}


\title{Scalability and Total Recall with Fast CoveringLSH
\titlenote{The research leading to these results has received funding from the European Research Council under the EU 7th Framework Programme, ERC grant agreement no.~614331.}
}

\numberofauthors{2} 
\author{
\alignauthor
Ninh Pham\\
       \affaddr{IT University of Copenhagen}\\
       \affaddr{Denmark}\\
       \email{ndap@itu.dk}
\alignauthor
Rasmus Pagh\\
       \affaddr{IT University of Copenhagen}\\
       \affaddr{Denmark}\\
       \email{pagh@itu.dk}
}

\maketitle

\begin{abstract}

Locality-sensitive hashing (LSH) has emerged as the dominant algorithmic technique for similarity search with strong performance guarantees in high-dimensional spaces. A drawback of traditional LSH schemes is that they may have \emph{false negatives}, i.e., the recall is less than 100\%. This limits the applicability of LSH in settings requiring precise performance guarantees. 
Building on the recent theoretical ``CoveringLSH'' construction that eliminates false negatives, we propose a fast and practical covering LSH scheme for Hamming space called \emph{Fast CoveringLSH (fcLSH)}.
Inheriting the design benefits of CoveringLSH our method avoids false negatives and always reports all near neighbors.
Compared to CoveringLSH we achieve an asymptotic improvement to the hash function computation time from $\BO{dL}$ to $\BO{d + L\log{L}}$, where $d$ is the dimensionality of data and $L$ is the number of hash tables. 
Our experiments on synthetic and real-world data sets demonstrate that \emph{fcLSH} is comparable (and often superior) to traditional hashing-based approaches for search radius up to 20 in high-dimensional Hamming space.

\end{abstract}




\section{Introduction}

Similarity search is a fundamental ingredient in algorithms for a wide range of computer applications, including machine learning, database management, information retrieval, and pattern recognition and analysis. 
This problem has become increasingly important and challenging in the era of big data since the use of computational resources such as storage and power becomes critical. \full{For instance, a typical search engine needs to crawl and index billions of web pages which accumulate to a multi-terabyte database~\cite{Manku_WWW07}.}
Content-based image retrieval systems now have to answer similarity queries over billion-size image databases~\cite{Zhang_TOMM13}. Large-scale collaborative filtering engines have to deal with tens of millions users' data~\cite{Das_WWW07}. The emergence of big data adds to both research and commercial applications the challenges of \textit{scale} and \textit{accuracy} for efficient similarity search.

In most such applications data can be represented or approximated as high-dimensional binary vectors, and Hamming distance is used as a similarity measure. For instance, a near-duplicate detection system uses hashing techniques\full{~\cite{Charikar_STOC02, Li_WWW10, Mitzenmacher_WWW14}}\short{~\cite{Charikar_STOC02, Li_WWW10}} to represent documents as binary vectors, and identifies them as near-duplicates if their Hamming distances are smaller than a threshold radius. In content-based image retrieval systems, a standard approach is to learn short binary codes to represent image objects such that the Hamming distance between codes reflects their neighborhood or semantic similarity in the original space\full{~\cite{Jegou_PAMI11, Salakhutdinov_IJAR09, Torralba_CVPR08, Weiss_NIPS08}}\short{, see e.g.~\cite{Salakhutdinov_IJAR09} and its references}. Retrieving similar images can be efficiently done by simply returning all images with codes within a small Hamming distance of the code of the query image.

Similarity search in Hamming space dates back to Minsky and Papert~\cite{Minsky_Papert}, who referred to it as the \emph{approximate dictionary} problem. The generalization to arbitrary spaces is now known as \textit{near neighbor search}. Due to the ``curse of dimensionality'', the performance of indexing techniques based on data or space partitioning generally degrades as dimensionality increases, and is eventually no better than a simple linear search~\cite{Weber_VLDB98}. \full{This poses a problem of scale for near neighbor search in applications dealing with a very large number of bit strings that might not even fit in the main memory of one machine.}

Since~1998 \textit{locality-sensitive hashing} (LSH)~\cite{Indyk_STOC98} has emerged as a basic primitive for near neighbor search in high-dimensional space. It alleviates the effects of the ``curse of dimensionality'' by considering an \textit{approximate} variant, and obtains sub-linear time for the approximation problem.\full{ In a nutshell, LSH hashes similar points into the same bucket with with high probability, and increases the gap between collision probability of similar and dissimilar points. The search candidates are data points that are hashed into the same bucket as the query point.
}
\full{Since its first introduction, several LSH schemes~\cite{Broder_STOC98, Charikar_STOC02, Datar_SOCG04, Indyk_STOC98, Li_WWW10, Weiss_NIPS08} and efficient LSH-based methods for near neighbor search~\cite{Bawa_WWW05, Dasgupta_KDD11, Gan_SIGMOD12, Gao_KDD15, Gionis_VLDB99, Liu_VLDB14, Lv_VLDB07, Panigrahy_SODA05, Satuluri_VLDB12, Shrivastava_NIPS14, Sundaram_VLDB13, Tao_SIGMOD09} have been proposed for a wide range of distance functions in high-dimensional space.}
\short{Since its first introduction, many LSH-based methods for near neighbor search have been proposed. We refer to~\cite{Gao_KDD15,Shrivastava_NIPS14} for an up-to-date overview of  developments.}
However, a drawback of classical LSH-based methods is the probabilistic guarantees that result in false negatives (i.e., the recall is below 100\%). This limits 
applicability of LSH in settings requiring high accuracy or precise performance guarantees, e.g., fingerprint recognition, entity resolution, and plagiarism detection.

Although the requirement of perfect recall ratio has not often been the primary focus when studying similarity search in Hamming space, there are many applications where this setting is relevant. For the problem of large-scale image search and recognition, learning binary codes for images to preserve their neighborhood or semantic similarity\full{~\cite{Norouzi_NIPS12, Salakhutdinov_IJAR09, Torralba_CVPR08}}\short{ (see~\cite{Norouzi_NIPS12} and its references)} is widely used due to the simplicity of the representation and fast query processing. 
False negative findings in querying a binary code can degrade the performance of classification and retrieval tasks. In fact, such methods often perform brute-force search for answering near neighbor queries. Recently researchers have found that the binary codes must be long enough (hundreds of bits) to preserve discrimination power and to achieve good performance\full{~\cite{Gong_CVPR13, Perronnin_CVPR10, Yu_ICML14}}\short{, see~\cite{Yu_ICML14} and its references}. As such applications arise in large-scale image data sets, the problem of scaling up similarity search in high-dimensional Hamming space is getting more important and more challenging. 

In a recent theoretical study, \textit{CoveringLSH}~\cite{Pagh_SODA16} was  proposed to address the issue of false negatives in LSH for Hamming space. 
Instead of independently selecting bit positions from high-dimensional binary vectors as the classic LSH method~\cite{Indyk_STOC98}, CoveringLSH carefully chooses correlated bit positions that ``cover'' all possible positions of $r$ differences, and thus eliminate false negatives. 
\full{To explore the practicality of this approach, we implemented the CoveringLSH construction and carried out an experimental study. We found that although the method can avoid false negatives and match the asymptotic complexity bound of classical LSH~\cite{Indyk_STOC98}, substantial practical improvements are possible.}
An issue of CoveringLSH is that it requires the computation of $L$ hash values of $d$ bits, where $d$ is the dimensionality of data and $L$ is the number of hash tables.
This becomes a bottleneck for large dimensions, since evaluation time proportional to $dL$ is unavoidable.

{\bf Fast CoveringLSH.} 
This paper presents \emph{Fast CoveringLSH} (fcLSH), a fast and practical evolution of CoveringLSH that scales much better to high dimensions.
Inheriting the design benefits of CoveringLSH, fcLSH can not only answer \textit{approximate} near neighbor search with provable sub-linear guarantees, but also report the \textit{exact} set of all near neighbors. 
\full{Our method is the first practical solution, to the best of our knowledge, to bridge the gap between approximate computation and exact results for similarity search in high-dimensional Hamming space.}
In addition, for low dimensions where $d \leq L$, fcLSH achieves higher precision  than CoveringLSH. 
Our experiments on synthetic and real-world data sets demonstrate that fcLSH is comparable and often superior to traditional hashing-based approaches for search radius up to~20 in high-dimensional Hamming space.

{\bf Technical contributions.} Observe that for $d\gg \log(n)$ we can decrease the size of the hash values from $d$ to $\BO{\log n}$ bits each, while not significantly changing collision probabilities, by applying universal hashing~\cite{Carter_STOC77}.
In order to avoid intermediate results of $dL$ bits we show how to \emph{interweave} a carefully chosen universal hash function with the Fast Hadamard Transform, such that $L$ hash values of $\BO{\log n}$ bits are computed directly.
Since the Hadamard matrix is related to the projection family used by CoveringLSH, the values computed in this way are identical to those obtained by hashing the $d$-bit hash values to $\BO{\log n}$ bits.
This approach achieves an asymptotic improvement to hash function computation time from $\BO{dL}$ to $\BO{d + L\log{L}}$, \full{where $d$ is the dimensionality of data and $L$ is the number of hash tables}\short{for $d$ dimensions and $L$ hash tables}. 

\full{
The organization of the paper is as follows. In Section~\ref{sec:preliminaries}, we describe background and preliminaries, including near neighbor search problems, an overview of LSH, the very recent CoveringLSH scheme, and some background on Hadamard codes. The proposed approach is presented and analyzed in Section~\ref{sec:algorithm}. In Section~\ref{sec:experiment}, we show experimental evaluations of our proposed approach on both synthetic and real-world data sets. Section~\ref{sec:relatedwork} briefly reviews related work. Section~\ref{sec:conclusion} summarizes the paper and presents research directions concerned with CoveringLSH scheme.}

\section{Background and Preliminaries}\label{sec:preliminaries}
\subsection{Problem Setting}

We study the problem of near neighbor search in Hamming space under Hamming distance. Due to the ``curse of dimensionality'', many proposed solutions for \textit{exact} near neighbor search in high-dimensional space become slower than simple linear search. In order to trade precision for speed, approximate versions of near neighbor search have been widely investigated in the literature, and locality-sensitive hashing-based methods have emerged as the most widely used solutions for such problems. The first approximate version, called $c$-approximate $r$-near neighbor search, is defined as follows.

\begin{definition}($c$-approximate $r$-near neighbor or $\left(c, r\right)$-NN)
Given a set $S \subset \left\{ 0, 1\right\}^d$, $|S| = n$, the Hamming distance function $d$, and parameters $r > 0$, $c > 1$, $\delta > 0$, construct a data structure such that, given any query $\bm{q} \in \left\{ 0, 1\right\}^d$, if there exists a point $\bm{x} \in S$ and $d(\bm{x}, \bm{q}) \leq r$, it reports some point $\bm{y} \in S$ where $d(\bm{y}, \bm{q}) \leq cr$ with probability $1 - \delta$. 
\end{definition}

We note that $(c, r)$-NN problem has two approximation factors, consisting of the approximation of distance by a factor of $c$, and the approximation of the result set determined by the success probability $1 - \delta$. Due to the approximation of distance, this problem formulation may give undesirable quality of results. By setting $c = 1$, the second approximate version, called $r$-near neighbor reporting, has more practical applications\full{~\cite{Salakhutdinov_IJAR09, Satuluri_VLDB12, Slaney_IEEE12, Sundaram_VLDB13}}\short{~\cite{Salakhutdinov_IJAR09, Sundaram_VLDB13}} since it reports \textit{all} points within distance~$r$ to the query.

\begin{definition}($r$-near neighbor reporting or $r$-NN)
Given a set $S \subset \left\{ 0, 1\right\}^d$, $|S| = n$, the Hamming distance function $d$, and parameters $r > 0$, $\delta > 0$, construct a data structure that, given any query $\bm{q} \in \left\{ 0, 1\right\}^d$, return each point $\bm{x} \in S$ where $d(\bm{x}, \bm{q}) \leq r$ with probability $1 - \delta$. 
\end{definition}

We call this the ``exact'' $r$-NN problem in case $\delta = 0$, otherwise it is the ``inexact'' $r$-NN to distinguish with the approximation term of $(c,r)$-NN problem. Note that the inexact factor here is due to reporting \textit{each} near neighbors, determined by the success probability $1 - \delta$. 

This work investigates the possibility of an \textit{exact} guarantee for $r$-NN problem in order to report every point $\bm{x} \in S$ where $d(\bm{x}, \bm{q}) \leq r$. It is worth noting that solving the exact $r$-NN problem implies an exact solution to the nearest neighbor problem with comparable performance by building several solutions for different radii~\cite{Andoni_JACM08}. For convenience of notation, we are now using a bold letter for a binary vector (e.g.,~$\bm{v}$) to distinguish it from a scalar quantity (e.g.,~$v$). In what follows, we talk about \textit{near} points at distance at most $r$ (those that should be reported), \textit{c-near} points at distance between $r$ and $cr$, an \textit{far} points at distance larger than $cr$.

\subsection{Locality-sensitive Hashing Functions}

LSH is one of the most widely used approaches to near neighbor search in high-dimensional space because it is able to break the $\BO{n}$  barrier for the $\left(c, r \right)$-NN problem.
\begin{definition}\emph{(Indyk and Motwani~\cite{Indyk_STOC98})}
\label{def:LSH}
Fix a distance function $d: \mathbb{U}\times \mathbb{U} \rightarrow {\bf R}$.
For positive reals $r$, $c$, $p_1$, $p_2$, and $p_1 > p_2$, $c > 1$, a family of functions $\mathcal{H}$ is \emph{$(r,cr,p_1,p_2)$-sensitive} if for uniformly chosen $h \in \mathcal{H}$ and all $\bm{x}, \bm{y}\in \mathbb{U}$:
\begin{itemize}[leftmargin=4ex]
	\item If $d(\bm{x}, \bm{y}) \leq r$ then $\Pr{h(\bm{x})=h(\bm{y})} \geq p_1$;
	\item If $d(\bm{x}, \bm{y}) \geq cr$ then $\Pr{h(\bm{x})=h(\bm{y})} \leq p_2$.
\end{itemize}
\end{definition}

The classic LSH family for Hamming distance uses a bit sampling approach~\cite{Gionis_VLDB99, Indyk_STOC98}. It is simply the family of all the projections of points to one dimension, i.e., a hash function value is just a random bit sample. That is, given a point $\bm{x} = \{x_1, \ldots , x_d \}$, the bit sampling LSH family $\mathcal{B}$ with parameters $p_1 = 1 - r/d$, $p_2 = 1 - cr/d$ is constructed as:
%
$$\mathcal{B} = \left\{h : \{ 0, 1\}^d \rightarrow \{0, 1\} \, | \, h(\bm{x}) = x_i \text{ for some } i \in \{1, \ldots, d\} \, \right\}.$$
%

The performance of LSH-based algorithms is governed by the parameter $\rho = \log{p_1}/\log{p_2}$, and constructing an LSH family with small $\rho$ automatically leads to the improved algorithms for the $\left(c, r \right)$-NN problem. For the bit sampling family $\mathcal{B}$, $\rho \approx 1/c$ which is optimal for data-independent LSH in Hamming space~\cite{Donnell_TOCT14}. 

The classical LSH-based algorithm for near neighbor search problem is as follows. \full{We concatenate $k$ random hash values to increase the gap of collision probability between near points and far points, and independently repeat the process $L$ times to increase the success probability of the algorithm. In particular, given an LSH family $\mathcal{H}$,}\short{Given an LSH family $\mathcal{H}$,} construct $L$ hash tables by hashing data points using $L$ hash functions $g_j$, $j = 1, \ldots, L$, by setting $g_j = \left( h_{j}^{1}, \ldots, h_{j}^{k} \right)$, where $h_j^{i}$, $i = 1, \ldots, k$, are chosen randomly from the LSH family $\mathcal{H}$. To process a query $\bm{q}$, one needs to retrieve candidate points from the bucket $g_j(\bm{q})$ in the $j$th hash table, $j = 1, \ldots, L$. For the candidate set retrieved, a filtering procedure is performed to remove false positives. There are different filtering strategies corresponding to $(c,r)$-NN and $r$-NN problems~\cite{Andoni_JACM08}.
%

%

{\bf Strategy~1:} Stop searching after finding the first $3L$ points (including duplicates) and return the point with minimum distance to the query $\bm{q}$.

{\bf Strategy~2:} For each distinct point $\bm{x}$ from the candidate set, compute $d(\bm{x}, \bm{q})$ and report $\bm{x}$ if $d(\bm{x}, \bm{q}) \leq r$.

Strategy~1 that interrupts the search after retrieving $3L$ points (including duplicates) is of significant importance in theory because it introduces a sub-linear time algorithm with suitable choices of $k$ and $L$ for the $\left(c, r \right)$-NN problem~\cite{Gionis_VLDB99, Indyk_STOC98}. In particular, it runs in $\BO{n^\rho}$ time where $\rho = \log{p_1}/\log{p_2}$ if we suitably choose $k = \BO{\log{n}}$, $L = \BO{n^\rho}$\full{, and interrupt the searching process after retrieving the first $3L$ points. Using the bit sampling family $\mathcal{B}$, it solves the $\left(c, r \right)$-NN problem in sub-linear time $\BO{n^{1/c}}$ using $\BO{n^{1 + 1/c}}$ space. Despite of the attractive asymptotic space and query performance, Strategy~1 may give undesirable quality of results compared to Strategy~2}.

Strategy~2 enables us to solve the $r$-NN problem, which has more practical applications\full{~\cite{Das_WWW07, Henzinger_SIGIR06, Manku_WWW07, Satuluri_VLDB12, Sundaram_VLDB13}}. It provides better result quality since all reported points are within distance $r$ to the query point. It might run in $\BO{n}$ time in the worst case, but for many natural data sets, proper settings of $k$ and $L$ still result in a sub-linear query time~\cite{Andoni_JACM08}. However, Strategy~2 can introduce false negatives if some near points do not collide with the query under any hash function. That limits the use of LSH in applications requiring high accuracy or precise performance guarantees.

For practical implementation\footnote{E2LSH. http://www.mit.edu/$\sim$andoni/LSH/\full{, \newline OptimalLSH. https://github.com/yahoo/Optimal-LSH}.}, the value $k = \BO{\log{n}}$ is large in a typical setting. One can reduce the time of checking collision and the amount of memory for bucket identification from $\BO{k}$ to $\BO{1}$ by using an associated universal hash function to hash a $k$-bit hash value into an integer. Moreover, since the domain of the hash function $g_j$ is too large to store all possible buckets explicitly, and we only need to store non-empty buckets, we use a hash table to contain these non-empty buckets. Given a prime $P$ and random integers $b_i$, $i = 1, \ldots, k$, from the interval $\left\{0, \ldots, P - 1 \right\}$, we use hash functions of the form:  
%
\begin{eqnarray}  
\label{eq:universal} 
p(x_1, \ldots, x_k) = b_1 \cdot x_1 + \cdots + b_k \cdot x_k \mod P.
\end{eqnarray}
According to~\cite{Carter_STOC77}, this family is universal which means that the probability of collision is small if $P$ is sufficiently large\full{ (say, $P > n^2$ when hashing a set of $n$ vectors)}.

\subsection{CoveringLSH}\label{sec:coveringLSH}

In very recent work~\cite{Pagh_SODA16}, a novel LSH scheme was proposed to solve the exact $r$-NN problem. This method always introduces a collision for every pair of binary vectors within a given radius $r$. Instead of independently selecting bit positions as in the bit sampling approach, CoveringLSH carefully chooses correlated bit positions so that it can ``cover'' \textit{all} possible positions of $r$ differences, which implies an \textit{exact} guarantee for the $r$-NN problem when used with Strategy~2. \full{The underlying LSH definition is as follows.}
\begin{definition}
\label{def:covering}
An LSH family $\mathcal{A}$ is \emph{$r$-covering} if for every two binary vectors $\bm{x}, \bm{y} \in \{0,1\}^d$ with Hamming distance $d(\bm{x}, \bm{y}) \leq r$, there exists $g \in \mathcal{A}$ such that $g(\bm{x}) = g(\bm{y})$.
\end{definition}

The proposed scheme relies on a \textit{random} mapping $m : [d] \rightarrow \{ 0, 1\}^{r+1}$ that maps bit positions to binary vectors of length $r + 1$. This $r$-covering LSH family, $\mathcal{A}$, consists of $2^{r+1} - 1$ correlated hash functions via the mapping $m$. Each hash function is associated with a binary vector of length $d$, denoted by $\bm{g}_v$, $v = 1, \ldots, 2^{r+1} - 1$ of the form
%
\begin{eqnarray}  
\label{eq:naive} 
\bm{g}_v = \left( \left\langle m(1), \bm{v} \right\rangle, \left\langle  m(2), \bm{v} \right\rangle, \cdots, \left\langle m(d), \bm{v} \right\rangle \right),
\end{eqnarray}
where $\left\langle m(i), \bm{v} \right\rangle = \sum_{j=1}^{r+1} m(i)_j v_j \mod 2$ is the dot product modulo 2 of two vectors $m(i)$ and $\bm{v}$. The hash value of a given binary vector $\bm{x}$ is simply the binary vector produced by the bit-wise \textbf{AND} operation, i.e.,~$g_v(\bm{x}) = \bm{g}_v \wedge \bm{x}$. The $2^{r+1} - 1$ hash functions of $\mathcal{A}$ correspond to all distinct non-zero binary vectors $\bm{v} \in \{ 0, 1 \}^{r+1} \setminus \{ \bm{0} \}$ or equivalently binary representations of $v \in \{1, \ldots, 2^{r+1} - 1\}$. Hence, the non-zero binary vector $\bm{v}$ or the corresponding integer $v$ is used to index the $v$th hash function $\bm{g}_v$, and we will use them interchangeably.

\begin{exmp}
Given the two binary vectors $\bm{x} = \underline{0}01\underline{1}$ and $\bm{q} = 1010$, we have that $d(\bm{x}, \bm{q}) = 2$. A 2-covering LSH family uses a random mapping $m : [4] \rightarrow \{ 0, 1\}^{3}$, e.g., $m(1) = 011, m(2) = 100, m(3) = 101, m(4) = 001$, to construct~7 hash functions as follows:
{\small
\begin{displaymath} 
\begin{aligned} 
\bm{g}_1 &= \left( \left\langle m(1), 001 \right\rangle, \left\langle  m(2), 001 \right\rangle, \left\langle m(3), 001 \right\rangle, \left\langle m(4), 001 \right\rangle \right) = 1011, \\
\bm{g}_2 &= \left( \left\langle m(1), 010 \right\rangle, \left\langle  m(2), 010 \right\rangle, \left\langle m(3), 010 \right\rangle, \left\langle m(4), 010 \right\rangle \right) = 1000, \\
\bm{g}_3 &= 0011, \bm{g}_4 = 0110, \bm{g}_5 = 1101, \bm{g}_6 = 1110, \bm{g}_7 = 0101.
\end{aligned}
\end{displaymath}
}
There is one collision between $\bm{x}$ and $\bm{q}$ corresponding to $\bm{g}_4 = 0110$ since $\bm{g}_4 \wedge \bm{x} = \bm{g}_4 \wedge \bm{q} = 0010$. Note that the 2-covering LSH family can cover all possible positions of~2 differences in~4-dimensional Hamming space.
\end{exmp}

%
\begin{theorem}\emph{~\cite[Lemma 3.2]{Pagh_SODA16}}
\label{thm:covering}
For every mapping $m : [d] \rightarrow \{0, 1\}^{r+1}$, the family $\mathcal{A}$ built as above is $r$-covering.
\end{theorem}

The new $r$-covering LSH scheme can not only eliminate the problem of false negatives but also essentially \textit{match} the complexity bound of the seminal LSH construction of Indyk and Motwani~\cite{Indyk_STOC98} if $cr = \log{n}$. This is due to Theorem~\ref{thm:collision}.
\begin{theorem}\emph{~\cite[Theorem 3.1]{Pagh_SODA16}}
\label{thm:collision}
For any two binary vectors $\bm{x}, \bm{y} \in \{0, 1 \}^d$ and a random mapping $m:[d] \rightarrow \{0, 1\}^{r+1}$, $\mathcal{A}$ has two following properties:
\begin{enumerate}[leftmargin=4ex]
	\item If $d(\bm{x}, \bm{y}) \leq r$ then $\Pr {\exists g \in \mathcal{A} : g(\bm{x}) = g(\bm{y})} = 1$.
	\item $\E{\left| \left\{ g \in \mathcal{A} \, | \, g(\bm{x}) = g(\bm{y}) \right\}\right|} < 2^{r+1 - d(\bm{x}, \bm{y})}$.
\end{enumerate}
\end{theorem}

It is obvious that, for the setting where $cr = \log{n}$, the number of hash functions is $2^{r+1} - 1 \approx 2n^{1/c}$ and the total expected number of collisions for the far points among all hash functions is at most $2^{r} \approx n^{1/c}$. This implies an efficient sub-linear algorithm for solving the $(c, r)$-NN problem with constant success probability, like the classic LSH schemes. In addition, the $r$-covering LSH scheme with Strategy~2 will answer the $r$-NN problem with an \textit{exact} guarantee, returning \emph{all} points within distance $r$ to the query. \full{Since the constraint $cr = \log{n}$ is a key requirement of $r$-covering LSH schemes, the next section will introduce generalizations to satisfy this constraint.}

The basic $r$-covering scheme needs time $\BO{d}$ to construct one hash function (see Equation~(\ref{eq:naive})). In practice, the dimensionality of binary data can be high, e.g., documents, recommendation data sets. Also, an embedding process to Hamming space can require high dimensionality, e.g., embedding $\ell_1$-norm into Hamming space by a \textit{unary} representation~\cite{Gionis_VLDB99}, semantic hashing to embed images into Hamming space~\cite{Gong_CVPR13, Yu_ICML14}. This issue demands significant computational resources for computing $r$-covering hash codes.

\subsection{Hadamard Codes}\label{sec:hadamardcode}

The Hadamard code is an error-correcting code that enables efficient and reliable message transmission over noisy channels.\full{ The message is encoded by adding some redundant information such that, if a small part of the encoded message is corrupted, we are still able to correct it and recover the original message.} Here we aim at using Hadamard codes to construct LSH hash functions, and we will not use its error-correcting properties. Instead, we explain how to generate Hadamard codes over the binary alphabet $\{0, 1\}$ and how to leverage it to construct CoveringLSH hash functions.

Given a binary vector $\bm{v} \in \{0, 1\}^{r+1}$, the Hadamard code maps $\bm{v}$ into a binary vector $\text{Had}(\bm{v})$ of length $2^{r+1}$ using an encoding function $\text{Had} : \{0, 1\}^{r+1} \rightarrow \{0, 1\}^{2^{r+1}}$. In particular, $\text{Had}(\bm{v})$ is generated as follows:
\begin{eqnarray} \label{eq:code}
\text{Had}(\bm{v}) = \left( \left\langle a(0), \bm{v} \right\rangle, \left\langle  a(1), \bm{v} \right\rangle, \cdots, \left\langle a(2^{r+1} - 1), \bm{v} \right\rangle \right),
\end{eqnarray}
where $a(i)$, $i = 0, \ldots, 2^{r+1} - 1$, are \textit{all} possible binary vectors in $\{ 0, 1\}^{r+1}$, and $\left\langle a(i), \bm{v} \right\rangle$ is the dot product modulo 2 of two vectors $a(i)$ and $\bm{v}$.\full{ It is worth noting that the first bit of the Hadamard code corresponding to $\left\langle a(0), \bm{v} \right\rangle$ is not used in practice since $a(0) = \bm{0}$ and this bit is always zero.}

Consider the hash function vector $\bm{g}_v$ in Equation~(\ref{eq:naive}) and the Hadamard code $\text{Had}(\bm{v})$ in Equation~(\ref{eq:code}). It is observed in~\cite{Pagh_SODA16} that $\bm{g}_v$ can be seen as sampling a subset of $\text{Had}(\bm{v})$ since the random mapping $m$ is a subset of $\{0, 1\}^{r+1}$. We note that the Hadamard code of a binary vector $\bm{v}$ corresponds to the $v$th row of the so-called \emph{Hadamard matrix} ${\bf H}$  of the same size using the mapping $1 \mapsto -1$ and $0 \mapsto 1$. Conversely, we can use the Hadamard matrix of size $2^{r+1} \times 2^{r+1}$ with the reverse mapping as  Hadamard codes for vectors $\bm{v} \in \{ 0, 1\}^{r+1}$. The next section will exploit this relation and show how to use Hadamard codes and the fast Hadamard transform \texttt{FHT()} to efficiently construct $r$-covering LSH families.

\section{Algorithm}\label{sec:algorithm}


\full{\subsection{Description of fcLSH}\label{sec:description}}

\subsubsection{A typical case} 

Let us now present an example of a typical setting of image search where Hadamard codes can be used as $r$-covering LSH functions \textit{without} any modifications. Suppose that we have a set of binary vectors $S \subseteq \{0, 1 \}^8$. Given a query $\bm{q}$, we would like to find all points within distance $r = 2$ from ${\bm q}$.\full{ A 2-covering LSH family requires 7 hash functions to cover \textit{all} possible~2 differences between data points and query.} We generate Hadamard codes ${\bf C}_{7, 8}$ by using the rows of Hadamard matrix ${\bf H}_{8, 8}$ as described above, and remove the first row to avoid trivial collisions. We see that ${\bf C}_{7, 8}$ is a $2$-covering LSH family. (We \emph{will} in fact use the first column of the Hadamard matrix to simplify the fcLSH description and construction; in the practical implementation we later discard it due to its trivial collision.)
\renewcommand{\arraystretch}{0.8}
\begin{displaymath}
	{\bf C}_{7, 8} = \begin{pmatrix}
       0 & 1 & 0 & 1 & 0 & 1 & 0 & 1 \\[0.3em]
       0 & 0 & 1 & 1 & 0 & 0 & 1 & 1 \\[0.3em]
			 0 & 1 & 1 & 0 & 0 & 1 & 1 & 0 \\[0.3em]
			 0 & 0 & 0 & 0 & 1 & 1 & 1 & 1 \\[0.3em]
       0 & 1 & 0 & 1 & 1 & 0 & 1 & 0 \\[0.3em]
       0 & 0 & 1 & 1 & 1 & 1 & 0 & 0 \\[0.3em]
			 0 & 1 & 1 & 0 & 1 & 0 & 0 & 1			 
							\end{pmatrix}		
\end{displaymath}

It is obvious that there exists at least one collision for every pair of vectors within distance~2. This implies an algorithm for the $2$-NN problem without false negatives using the ${\bf C}_{7,8}$ LSH family (see Example~\ref{ex:2covering}). Note that, in this case, the mapping $m(i)$ of the $r$-covering LSH scheme is simply the vector representing $i$ in binary.
\begin{exmp}
\label{ex:2covering}
Given the two binary vectors $\bm{x} = 0011\underline{0}01\underline{1}$, $\bm{y} = 0011\underline{0}0\underline{01}$ and the query vector $\bm{q} = 00111010$, we have that $d(\bm{x}, \bm{q}) = 2$, $d(\bm{y}, \bm{q}) = 3$. Given the 2-covering LSH family ${\bf C}_{7,8}$, there is one collision between $\bm{x}$ and $\bm{q}$ (i.e., $g_3(\bm{x}) = g_3(\bm{q}) = 00100010$) corresponding to the 3rd row of ${\bf C}_{7,8}$, and there is no collision between $\bm{y}$ and $\bm{q}$.
\end{exmp}

In a typical setting for large-scale image search, suppose that we have a set $S$ of $n = 2^{24}$ vectors from $\{ 0, 1\}^{128}$. Given a query $\bm{q}$, we may wish to search all vectors in $S$ within distance $r = 6$ from $\bm{q}$. Since an exhaustive search in Hamming balls with $r = 6$ would take much more time than just linear search, we settle for a 4-approximate similarity search ($cr = \log{n}$). 

The $6$-covering LSH requires~127 hash functions and we use a random column-based permutation of the Hadamard codes ${\bf C}_{127, 128}$ as the LSH family. Theorem~\ref{thm:collision} shows that near vectors within radius~6 always collide with $\bm{q}$ in at least~1 hash function. Moreover, in expectation, a far-away vector at distance larger than~$24$ has collision probability at most~$1/2^{24} = 1/n$ under each hash function. This means that the $r$-covering LSH scheme can be used for efficiently answering the \textit{exact} $r$-NN search by pruning almost all far vectors.

\subsubsection{The general case}

We now consider the general case of $r$-covering LSH schemes for answering exact $r$-NN queries. Note that the constraint $cr = \log{n}$ affects the efficiency of $r$-covering LSH-based algorithms because it determines the pruning power. Moreover, it is also the key factor governing the ``best tradeoff'' between space and time complexity for the $(c, r$)-NN problem. Another hurdle for $r$-covering LSH schemes is that high dimensionality $d$ requires significant hash function computation time.


\full{Keep in mind that the number of hash tables of the $r$-covering scheme and the classical scheme are $2^{r+1} - 1$ and $\BO{n^{1/c}}$, respectively. The total expected number of collisions for far points with $r$-covering LSH is at most $n2^{r}/2^{cr}$, whereas that of the classic scheme is $\BO{n^{1/c}}$. So it is clear that when $cr = \log{n}$, both approaches have the same time complexity and space usage for near neighbor search.}

We use a method from~\cite{Pagh_SODA16} to handle the constraint $cr = \log{n}$. \full{It is clear that when $cr < \log{n}$, the number of hash tables is smaller but the number of collisions is larger than for the classic LSH scheme. Intuitively, we need to increase the radius $r$ by simply replicating the dimensionality of both data and query points $\left\lfloor \log{n}/cr \right\rfloor$ times (see Example~\ref{ex:preprocessing}). On the other hand, when $cr > \log{n}$, the space usage for hash tables is larger but the number of collisions is smaller than for the classic LSH scheme. In order to reduce the radius $r$ while still maintaining the exactness guarantee, we leverage the pigeonhole principle by first permuting and then partitioning the dimensions of both data and query points into $\left\lceil cr/ \log{n} \right\rceil $ parts (see Example~\ref{ex:preprocessing}). Then we independently build LSH data structures for each partition and candidate vectors are generated for each partition.} 
\short{When $cr < \log{n}$, we increase the radius $r$ by replicating each bit of both data and query points $\left\lfloor \log{n}/cr \right\rfloor$ times. On the other hand, when $cr > \log{n}$, we leverage the pigeonhole principle by first permuting and then partitioning the dimensions of both data and query points into $\left\lceil cr/ \log{n} \right\rceil $ parts (see Example~\ref{ex:preprocessing}). Then we independently build LSH data structures for each partition and the candidate vectors are generated for each partition.}
\begin{exmp}
\label{ex:preprocessing}
Given a binary vector $\bm{q} = 0011$, replicating $\bm{q}$ 2 times returns a new vector $\bm{q}^{(2)} = 00110011$. A random permutation of $\bm{q}$ gives $\bm{q}' = 0110$. Partitioning $\bm{q}'$ into 2 parts returns two vectors $\bm{q}_1 = 01, \bm{q}_2 = 10$.
\end{exmp}

After replicating or partitioning the dimensions, we use a new query radius $r'$ where $cr' \approx \log{n}$. Denote by $d'$ the new dimensionality of data and query points.  
If $d' > 2^{r'+1}$, we will need a random mapping $m: [d'] \rightarrow [2^{r'+1}]$ that randomly samples $d'$ columns from the Hadamard codes ${\bf C}_{2^{r'+1}-1, 2^{r'+1}}$ to form the $r'$-covering LSH family. On the other hand, if $d' \leq 2^{r'+1}$, we can leverage a 0-padding trick to increase the dimensionality to $2^{r'+1}$ without changing~$r'$, and simply use the Hadamard codes ${\bf C}_{2^{r'+1}-1, 2^{r'+1}}$ with columns randomly permuted as the $r'$-covering LSH family, as in the typical case above. In both cases, if $d'$ is large, it affects the hashing cost, i.e., computing the hash value and identifying the bucket corresponding to the query. The trick of converting long binary hash values into integers, see the Equation~(\ref{eq:universal}), still requires $\BO{d'2^{r'+1}}$ time. To address this problem, we propose to use the fast Hadamard transform for quickly computing integer hash values in $\BO{d' + r'2^{r'+1}}$ time, which is asymptotically faster when $d' > r'$.

\subsection{Construction}\label{sec:construction}

As elaborated above, we need to satisfy the constraint $cr = \log{n}$ in order to achieve high pruning power like the classic LSH scheme. We handle this issue by simply replicating or partitioning the dimensionality of both data and query points to increase or decrease the radius $r$ to be approximately $\log{(n)}/c$, as illustrated in Algorithm~\ref{alg:preprocessing}. 
%
\begin{algorithm}[t]
\caption{Pre-processing algorithm}
\label{alg:preprocessing} 									
\begin{algorithmic} [1]
\REQUIRE {A vector $\bm{q} = \{q_1, \ldots, q_d \}$, radius $r > 0$, approximation ratio $c > 1$, and data set size $n$} 
\IF {$cr < \log{n}$}	
	\STATE {$\bm{q}$ is replicated $t = \left\lfloor \log{(n)}/cr \right\rfloor$ times to form a new vector
	$\bm{q}^{(t)} = \underbrace{\bm{q} \ldots \bm{q}}_{t \text{ times}}$}	
\ELSIF {$cr > \log{n}$}
	\STATE { Randomly permute $\bm{q}$ }
	\STATE {$\bm{q}$ is partitioned into $t = \left\lceil cr/\log{n} \right\rceil$ parts to form $t$ new vectors: \\
	$\bm{q}_1 = \{q_1, \ldots , q_{\left\lfloor d/t \right\rfloor}\}, \ldots, \bm{q}_t = \{q_{(t-1)\left\lfloor d/t \right\rfloor}, \ldots, q_d \}$}
\ENDIF
\normalsize
\end{algorithmic}
\end{algorithm}

For simplicity of notation, let us denote by $d$ and $r$ the new dimensionality of data and the new query radius, respectively, after pre-processing data to satisfy $cr \approx \log{n}$. We now present two variants of the fcLSH scheme: a general construction using a random mapping $m : [d] \rightarrow [2^{r+1}] $ for $d > 2^{r+1}$ as introduced in~\cite{Pagh_SODA16} and a specific construction using a random permutation $m : [2^{r+1}] \rightarrow [2^{r+1}]$ for $d \leq 2^{r+1}$. In both cases, we exploit the fast Hadamard transform for fast computation of hash functions.

\textbf{The general construction for \boldmath$d > 2^{r+1}$.}
Recall that the basic $r$-covering LSH family requires $L = 2^{r+1}-1$ hash functions and the construction of a hash function $g_v$ relies on a random mapping $m: [d] \rightarrow \{0, 1\}^{r+1}$ and dot products modulo 2 between $m(i)$ and $\bm{v}$, described in Equation~(\ref{eq:naive}). This procedure is identical to randomly sampling $d$ positions among $2^{r+1}$ positions from $\text{Had}(\bm{v})$, the Hadamard code of the vector $\bm{v}$. This implies that we can use a new random \emph{mapping} $m : [d] \rightarrow [2^{r+1}]$ and rely on a simple construction without computing $d$ dot products as follows.
\short{\vspace{-2mm}}
\begin{eqnarray} 
\label{eq:hadamardcode}
\bm{g}_v = \{ \text{Had}(\bm{v})_{m(1)}, \text{Had}(\bm{v})_{m(2)}, \ldots, \text{Had}(\bm{v})_{m(d)}\}.
\end{eqnarray}
\short{\vspace{-2mm}}

\textbf{The specific construction for \boldmath$d \leq 2^{r+1}$.}
It is obvious that any collision caused by the random mapping $m$ yields more collisions for both close points and far points. That might slightly degrade the performance of filtering mechanisms. In typical settings of content-based image retrieval applications where $d \leq 2^{r+1}$, we can combine the 0-padding trick with a random \emph{permutation} $m :[2^{r+1}] \rightarrow [2^{r+1}]$ over columns of the Hadamard codes ${\bf C}_{2^{r+1}-1, 2^{r+1}}$ to achieve better results than the construction in Equation~(\ref{eq:hadamardcode}). This idea is illustrated in Algorithm~\ref{alg:FHT} (lines 7--8).

\textbf{Fast computation of hash functions.} We use a conventional hash function to map a binary hash value of length $d$ into an integer hash value in order to reduce the amount of memory for bucket identification and time complexity of searching a bucket in a hash table. A na{\"i}ve approach \full{to convert~$L$ binary hash codes into $L$ integers asymptotically} requires $\BO{dL}$ time complexity, see Equation~(\ref{eq:universal}). We show that we can reduce this cost to $\BO{d + L\log{L}}$ by using the fast Hadamard transform \texttt{FHT()}\full{. The pseudocode in Algorithm~\ref{alg:FHT} shows how to efficiently construct the $r$-covering LSH family and compute hash values for any data point.}\short{, as illustrated in Algorithm~\ref{alg:FHT}.} 
\begin{algorithm}[t]
\caption{Generating hash values using the fast Hadamard transform \texttt{FHT()}}
\label{alg:FHT} 									
\begin{algorithmic} [1]
\REQUIRE {A point $\bm{q} \in \{0, 1 \}^d$, a prime $P$} 
\ENSURE {A vector ${\bm h}$ of $L = 2^{r+1}-1$ integer hash values}
\STATE {Pick a random integer-valued vector $\bm{b} \in [P]^{d}$}
\STATE {Compute a new integer-valued vector $\bm{\widetilde{q}} = \bm{q} * \bm{b}$ by component-wise multiplication}
\IF {$d > 2^{r+1}$}	
	\STATE {Pick a random mapping $m : [d] \rightarrow [2^{r+1}]$}		
	\STATE {Compute a sketch vector ${\bm t}$ where $t_j =\sum_{i:m(i)=j}{\bm{\widetilde{q}}_i} $}	
\ELSE
	\STATE {Pick a random permutation $m : [2^{r+1}] \rightarrow [2^{r+1}]$}		
	\STATE {${\bm t} := m(\bm{\widetilde{q}} \enspace ; \bm{0}^{2^{r+1} - d})$}	
\ENDIF	
	\STATE {${\bm h} = \frac{1}{2}\left( \left\|\bm{\widetilde{q}} \right\|_1 {\bf 1} - \texttt{FHT}(\bm{t}) \right) \mod P$}
	\STATE {Remove the first element from ${\bm h}$}
\normalsize
\end{algorithmic}
\end{algorithm}

\full{The algorithm for quickly generating hash values works as follows.} We generate a random seed vector~$\bm{b}$ to convert binary hash codes into integers (line~1) and compute a new vector $\bm{\widetilde{q}} = \bm{q} * \bm{b}$ by component-wise multiplication, i.e., $(\bm{q} * \bm{b})_i = q_i b_i$. 
Non-zero entries of $\bm{\widetilde{q}}$ correspond to 1s in $\bm{q}$. If the dimensionality of data is greater than the length of Hadamard codes, i.e.,~$2^{r+1}$, we evaluate the random mapping~$m$ on each dimension~$i$ of $\bm{\widetilde{q}}$, and sum up colliding entries to form the new sketch vector $\bm{t}$ of length $2^{r+1}$ (line~5). Otherwise, we apply 0-padding trick on $\bm{\widetilde{q}}$ and randomly permute it to get $\bm{t}$ (line~8). We note that applying a random permutation on $\bm{q}$ is equivalent to applying a random permutation on the Hadamard codes, because we are only concerned with collisions. $\texttt{FHT}(\bm{t})$ is then used to reduce the cost of computing $2^{r+1}$ integer hash values (line~10). Finally, we ignore the first element corresponding to the first row of the Hadamard matrix to get $L = 2^{r+1}-1$ integer hash values (line~11).
 
\full{Section~\ref{sec:analysis} will present our theoretical analysis of the correctness of Algorithm~\ref{alg:FHT}. It also shows that the fcLSH scheme provided by Algorithm~\ref{alg:FHT} is an efficient $r$-covering LSH scheme for near neighbor search problems.}

\textbf{Time complexity analysis.} 
\full{We now analyze the time complexity of the Algorithm~\ref{alg:FHT}.} Denote by $nnz({\bm q})$ the number of non-zero entries of vector $\bm{q}$. The running time at line~11 using the fast Hadamard transform \texttt{FHT()} is $\BO{L\log{L}}$. The other computational costs are bounded by $\BO{nnz({\bm q})}$. The total running time is $\BO{nnz({\bm q}) + L\log{L}}$, which can be compared to $\BO{nnz({\bm q})L}$ of the basic $r$-covering LSH scheme~\cite{Pagh_SODA16}. When $nnz({\bm q}) > \log{L}$, fcLSH is sufficiently faster than the basic $r$-covering scheme.


\subsection{Theoretical Analysis}\label{sec:analysis}

Now we sketch a theoretical analysis of the correctness of fcLSH. We first show that the general and the specific construction are efficient $r$-covering LSH schemes with the two properties stated in Theorem~\ref{thm:collision}. 
Note that the first property guarantees that fcLSH always eliminates false negatives and reports all near neighbors for the $r$-NN problem. The second property says that fcLSH has the same pruning power as the classic LSH scheme~\cite{Indyk_STOC98}. 
Then we argue that Algorithm~\ref{alg:FHT} computes exactly the same results as using the universal hash function $p()$ defined in Equation~(\ref{eq:universal}) to convert a binary hash value into an integer.\full{ As a consequence, we prove that fcLSH with the pre-processing steps in Algorithm~\ref{alg:preprocessing} (replicating or partitioning) is also an efficient $r$-covering LSH scheme for near neighbor search. } 

\full{
The following lemmas show that both the general and the specific construction gives good $r$-covering LSH families.
\begin{lemma}\label{lem:random}
Given a random mapping $m: [d] \rightarrow [2^{r+1}]$, an $r$-covering LSH family $\mathcal{A}$ can be constructed by selecting $d$ columns $m(1), \ldots, m(d)$ from the Hadamard codes ${\bf C}_{2^{r+1}-1, 2^{r+1}}$. This family satisfies properties 1 and 2 in Theorem~\ref{thm:collision}.
\end{lemma} 
\begin{proof}
The proof is straightforward since the procedure of randomly sampling $d$ columns from the Hadamard codes ${\bf C}_{2^{r+1}-1, 2^{r+1}}$ is identical to the basic construction of $r$-covering LSH in Equation~(\ref{eq:naive}). 
\end{proof}


It is worth noting that we can use a random mapping $m: [d] \rightarrow [2^{r+1}] \backslash \{1\}$ to ignore the first column $\bm{0}$ of the Hadamard codes. This mapping produces an $r$-covering LSH family $\mathcal{A'}$ with a sharper bound for the 2nd property of the $r$-covering scheme. Following up to the proof of~\cite[Theorem~3.1]{Pagh_SODA16}, we have:
\begin{displaymath}
\begin{aligned}
\E{\left| \left\{ g \in \mathcal{A'} \, | \, g(\bm{x}) = g(\bm{y}) \right\}\right|} &< 2^{r+1}(\frac{1}{2} - \frac{1}{2^{r+1}})^{d(\bm{x}, \bm{y})} \\
 &<2^{r+1 - d(\bm{x}, \bm{y})}.
\end{aligned}
\end{displaymath}
\begin{lemma}\label{lem:deterministic}
A random column-based permutation of Hadamard codes ${\bf C}_{2^{r+1}-1, 2^{r+1}}$ is an $r$-covering LSH family $\mathcal{A}$ for data sets with dimensionality $d \leq 2^{r+1}$. This family satisfies properties 1 and 2 in Theorem~\ref{thm:collision}.
\end{lemma}
\begin{proof}
This follows from the proof of~\cite[Theorem~3.1]{Pagh_SODA16}. However, we sketch the proof here for completeness. Given two near binary vectors $\bm{x}, \bm{y} \in \{0, 1\}^{d}$ with $d(\bm{x}, \bm{y}) \leq r$, let $\bm{z} = \bm{x}  \oplus \bm{y}$ satisfy $\left\|\bm{z}\right\|_1 \leq r$. It is clear that a collision between $\bm{x}$ and $\bm{y}$ under a hash function vector $\bm{g}$ corresponds to $ \bm{g} \wedge \bm{z} = \bm{0}$. In other words, any $r$ bit positions with 1s are mapped to zero under the hash function vector $\bm{g}$.

Since $d \leq 2^{r+1}$, according to Theorem~\ref{thm:covering}, for every random permutation $m: [2^{r+1}] \rightarrow [2^{r+1}]$, the construction shown in Equation~(\ref{eq:naive}) leads to an $r$-covering LSH scheme. This means that ${\bf C}_{2^{r+1}-1, 2^{r+1}}$ is $r$-covering.

In order to satisfy the second property, we need a random column-permutation of ${\bf C}_{2^{r+1}-1, 2^{r+1}}$. This trick will prevent the worst-case data sets where far pairs always collide due to $d < 2^{r+1}$. The random permutation of columns of ${\bf C}_{2^{r+1}-1, 2^{r+1}}$ makes the bit value $g_i \wedge z_i$ random, so the probability that $z_i = 1$ and $g_i \wedge z_i = 0$ is~$1/2$. It can be shown that the probability that $\bm{g} \wedge \bm{z} = \bm{0}$ is bounded by $2^{-\left\|\bm{z}\right\|_1}$. By linearity of expectation, summing over $2^{r+1}-1$ rows of ${\bf C}_{2^{r+1}-1, 2^{r+1}}$, the second property follows.
\end{proof}
}

\short{
\begin{lemma}\label{lem:rcovering}
Both general and specific constructions described above give $r$-covering LSH families. These families satisfy properties 1 and 2 in Theorem~\ref{thm:collision}.
\end{lemma}
\begin{proof}
Using the proofs of~\cite[Lemma~3.2]{Pagh_SODA16} and~\cite[Theorem~3.1]{Pagh_SODA16}, we can verify the claim. 
\end{proof}

}

It is worth noting that in a typical setting where $d = 2^{r+1}$, the size of ${\bf C}_{2^{r+1}-1, 2^{r+1}}$ is close to the smallest possible for an $r$-covering LSH family. Observe that we have $\binom{2^{r+1}}{r}$ possible sets of $r$ differences, and each row of ${\bf C}_{2^{r+1}-1, 2^{r+1}}$ can cover at most $\binom{2^{r}}{r}$ such sets. This means that the number of hash functions needed is at least $\binom{2^{r+1}}{r} / \binom{2^{r}}{r} > 2^r$, which is within a factor of 2 from the upper bound. This implies that the specific construction often gives better results than the general construction.

Next, we argue that Algorithm~\ref{alg:FHT} produces $r$-covering LSH families. Before presenting lemmas, let us describe the main technical insight used in Algorithm~\ref{alg:FHT}. Consider the ideal case where $d = 2^{r+1}$, and recall that the Hadamard code matrix ${\bf C}$ can be generated by the Hadamard matrix ${\bf H}$ with the same size by mapping $1 \mapsto 0$ and $-1 \mapsto 1$. If we let $\bf{1}$ denote the matrix with~1 in every entry, we have ${\bf C} = \left( \bf{1} - {\bf H} \right) / 2$. Given any binary vector $\bm{q}$, the hash value of $\bm{q}$ under the hash function vector ${\bf C}_v$ (the $v$th row of {\bf C}) is $g_v(\bm{q}) = {\bf C}_v \wedge \bm{q}$. Using the universal hash function $p()$ in Equation~(\ref{eq:universal}), we need a prime $P$ and a random seed vector $\bm{b}$ for computing $\bm{b} \cdot g_v(\bm{q}) \mod P$. This means that we need to compute the matrix-vector multiplication ${\bf C}\bm{\widetilde{q}}$, where $\bm{\widetilde{q}} = \bm{q} * \bm{b}$ is a component-wise product, as follows:
\short{\vspace{-2mm}}
\begin{eqnarray}  
\label{eq:FHT} 
{\bf C}\bm{\widetilde{q}} = \frac{1}{2}(\bm{1} - {\bf H})\bm{\bm{\widetilde{q}}} = \frac{1}{2}\left\|\bm{\bm{\widetilde{q}}}\right\|_1  {\bf 1} - \frac{1}{2}{\bf H}\bm{\bm{\widetilde{q}}}.
\end{eqnarray}

\begin{lemma}\label{lem:exactFHT}
Given a prime $P$ and any random seed vector $\bm{b} \in [P]^{d}$, Algorithm~\ref{alg:FHT} computes the same hash values as using $p()$ in Equation~(\ref{eq:universal}) on the $r$-covering LSH scheme introduced in~\cite{Pagh_SODA16}.
\end{lemma}
\begin{proof}
Since the random permutation used in the specific construction is a special case of the random mapping used in the general construction, we need only prove the claim for the general construction. 

Given any binary vector $\bm{q}$, we let $\bm{\widetilde{q}} = \bm{q} * \bm{b}$. It is clear that the contribution of $b_i q_i$ to the integer hash value is determined by the random mapping value $m(i)$. We form a vector $\bm{\xi}_i \in {\bf N}^{2^{r+1}}$ corresponding to the contribution of $b_i q_i$ whose entry at position $m(i)$ is $q_i b_i$ and the others are zero. The integer-value hash values of $\bm{q}$ is then computed as follows:
\begin{displaymath}
\begin{aligned}
{\bf C} \left( \sum_{i = 1}^{d}{\bm{\xi}_i} \right) &= \frac{1}{2} \sum_{i=1}^{d}{ \left( b_i q_i - {\bf H}\bm{\xi}_i \right) } = \frac{1}{2} \left( \sum_{i=1}^{d}{b_i q_i} - \sum_{i=1}^{d}{{\bf H}\bm{\xi}_i} \right)\\
&= \frac{1}{2} \sum_{i=1}^{d}{b_i q_i} - \frac{1}{2} {\bf H} \sum_{i=1}^{d}{\bm{\xi}_i} =  \frac{1}{2}\left\|\bm{\widetilde{q}}\right\|_1  {\bf 1} - \frac{1}{2}{\bf H}\bm{t},
\end{aligned}
\end{displaymath}
 where the vector ${\bm t}$ is computed by $t_j =\sum_{i:m(i)=j}{b_i q_i}$. Applying \texttt{FHT()} on the second term proves the claim.
\end{proof}

\begin{corollary}\label{lem:FHT_general}
Given a sufficiently large prime $P$, a construction of fcLSH provided by Algorithm~\ref{alg:FHT} is an $r$-covering LSH scheme with properties~1 and~2 of Theorem~\ref{thm:collision}.
\end{corollary}

\full{
Now we consider fcLSH with the pre-processing step in Algorithm~\ref{alg:preprocessing}. Due to the replication and partitioning step to satisfy $cr \approx \log{n}$, fcLSH does not have as strong a guarantee as the 2nd property in Theorem~\ref{thm:collision}. However, according to~\cite[Theorem~4.1]{Pagh_SODA16}, we derive the following extension of Theorem~\ref{thm:collision} for fcLSH.

\begin{lemma}
\label{lm:extension}
For any two binary vectors $\bm{x}, \bm{y} \in \{0, 1 \}^d$ and a random mapping $m:[d] \rightarrow \{0, 1\}^{r+1}$, an LSH family $\mathcal{A}$ constructed by fcLSH has following properties:
\begin{enumerate}[leftmargin=4ex]
	\item If $d(\bm{x}, \bm{y}) \leq r$ then $\Pr {\exists g \in \mathcal{A} : g(\bm{x}) = g(\bm{y})} = 1$.
	\item For the case requiring replication factor $t = \left\lfloor \log{(n)}/cr \right\rfloor$, \\
	$\E{\left| \left\{ g \in \mathcal{A} \, | \, g(\bm{x}) = g(\bm{y}) \right\}\right|} < 2n^{1/c}2^{-td(\bm{x}, \bm{y})}$.
	
	For the case requiring $t = \left\lceil \log{(n)}/cr \right\rceil$ partitions, \\
	$\E{\left| \left\{ g \in \mathcal{A} \, | \, g(\bm{x}) = g(\bm{y}) \right\}\right|} < 2n^{1/c}t\left( 1 - \frac{1}{2t} \right)^{d(\bm{x}, \bm{y})}$.
\end{enumerate}
\end{lemma}

From Lemma~\ref{lm:extension}, the performance of fcLSH with a partitioning step is slightly worse than classic LSH. This observation matches the theoretical analysis in~\cite{Pagh_SODA16}, which states that the $r$-covering scheme differs from the classic LSH scheme by at most a factor $\ln{(4)} < 1.4$ in the exponent for the general values of~$cr$.
}

\full{
\subsection{Discussion}\label{sec:discussion}

It is clear that for the problem of reporting all near neighbors, any algorithm may return many (or even all) points if a large fraction of the data set is close to the query point. This means that there is no sublinear guarantee on the running time of such algorithms. However, there are many natural data sets with the property that the distance gap between near points and far points is large. For these data sets, LSH-based approaches with their efficient pruning mechanism enable us to quickly report all near neighbors given a query point. We chose the Webspam data set\footnote{http://www.csie.ntu.edu.tw/$\sim$cjlin/libsvmtools/datasets} and applied the standard cosine similarity LSH~\cite{Charikar_STOC02} to each document to get data sets of 64-bit and 256-bit fingerprint vectors, respectively. 
\begin{figure} [t]
\centering
\includegraphics[width=1.0\columnwidth]{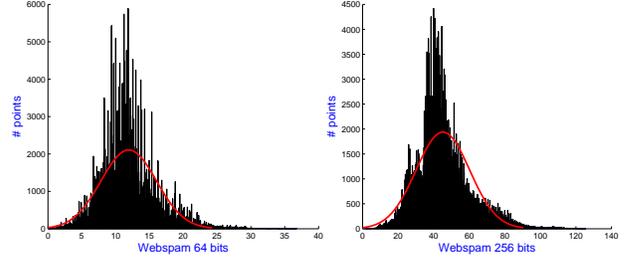}
\caption{The distance histogram between~50 random queries and a sampled set of data points of Webspam~64 bits (left) and Webspam~256 bits (right).}
\label{fig:Distribution}
\end{figure}

Figure~\ref{fig:Distribution} shows the distance histogram between~50 random queries and a sampled set of points of Webspam 64 bits and Webspam 256 bits. Given a vector query $\bm{q}$, we wish to search all vectors within distance $r = 10$ from $\bm{q}$. It is clear that any approach to answer this problem on Webspam 64 bits needs to return almost half number of points. LSH-based approaches with approximation ratio $c = 1.5$ are only able to filter away less than 30\% number of points. This means that LSH-based approaches might be outperformed by simple linear search. However, on the 256-bit version, the filtering mechanism of LSH-based approaches works efficiently. With an approximation ratio $c = 3$, LSH-based approaches can filter away up to 70\% of points. This implies a possible speedup of~3 times compared to linear search. In this setting, fcLSH clearly outperforms linear search while preserving the exactness guarantee. 

In general, LSH-based approaches are able to efficiently solve the problem of reporting all near neighbors for data sets that have large distance gap between near points and far points. For data sets that do not have such large distance gap, linear search might be a better choice.
}

\section{Experiment}\label{sec:experiment}

We implemented fcLSH in C++ and conducted experiments on an Intel Xeon Processor E5-1650~v3 with~64GB of RAM. We compared the performance of hashing-based algorithms for reporting all near neighbors, including our fcLSH scheme, the basic $r$-covering LSH~\cite{Pagh_SODA16}, the classic LSH scheme~\cite{Indyk_STOC98}, and the multi-index hashing approach~\cite{Norouzi_PAMI14} on synthetic and real-world data sets. Each result is the average of~5 runs over a query set of an algorithm.

\subsection{Experiment Setup}\label{sec:expsetting}

We consider alternative hashing-based approaches with performance guarantees in Hamming space for comparison. The following algorithms are used.
%
\begin{itemize}[leftmargin=4ex]
	\item \textbf{fcLSH}: Our method with fast computation of hash function using \texttt{FHT()}.	
	\item \textbf{bcLSH}: The basic covering construction~\cite{Pagh_SODA16} based on random samples from Hadamard codes.
	\item \textbf{LSH}: The classic LSH~\cite{Indyk_STOC98} using bit sampling approach.
	\item \textbf{MIH}: The recent multi-index hashing approach~\cite{Norouzi_PAMI14} running in sub-linear time for exact $r$-NN over uniformly distributed 
	data sets. 
\end{itemize}

Note that MIH is an alternative to exhaustive search in Hamming balls over data sub-dimensions. Based on the pigeonhole principle, MIH partitions data dimensions to reduce the radius, which is similar to our approach. However, the sub-linear guarantees of MIH is based on the strong assumption of uniform distribution of data points which is not true in many natural data sets\full{~\cite{Liu_ICDE11, Jegou_PAMI11, Manku_WWW07}}\short{~\cite{Jegou_PAMI11, Manku_WWW07}}.

\begin{table}[t]
\centering
\caption{Hash function computation time}
\label{tb:complexity} 
\begin{tabular}{|c|c|c|c|c|} \hline
Method & fcLSH & bcLSH & LSH & MIH  \\ \hline
Time & $\BO{d + L\log{L}}$ & $\BO{dL}$ &  $\BO{kL}$ &  $\BO{d}$ \\ \hline
\end{tabular}
\end{table}
\begin{table}[t]
\centering
\caption{Data set properties}
\label{tb:datasets} 
\begin{tabular}{|c|c|c|c|} \hline
Data sets & $n$ & $d$ & Binarization \\ \hline
ANN\_SIFT1M & 1M & 128 & LSH \\ \hline
Webspam & 0.35M & 254 & LSH \\ \hline
Enron &  $\sim$ 40K & $\sim$ 28K & Word freq. > 10 \\ \hline
MovieLens & $\sim$ 0.23M & $\sim$ 140K & Rating > 2 \\ \hline
\end{tabular}
\end{table}

\textbf{Parameter settings.} It is obvious that each hashing-based method achieves the best performance given the proper choices of parameters. Since such proper choices primarily depend on the distance distribution between queries and data points, we use suggested settings as below.
\full{
\begin{itemize}[leftmargin=4ex]
	\item For the general $r$-covering LSH schemes, including fcLSH and bcLSH~\cite{Pagh_SODA16}, we only need the partition trick when $r$ is large (say, $ r \geq 10$) since in that case we might not have enough space for $L = 2^{r+1}-1$ hash tables.	
	\item For classic LSH, we simply set the number of hash tables $L = 2^{r+1}-1$ for the sake of comparison. The number of bit samples is set as  $k = \left\lceil \log{(1 - \delta^{1/L})} / \log{(1-r/d)}\right\rceil$ where $\delta$ is the false negative ratio\footnote{http://www.mit.edu/$\sim$andoni/LSH/manual.pdf}.
	\item For MIH, the number of partitions is $\left\lceil d / \log_2{n}\right\rceil$ as suggested in~\cite{Norouzi_PAMI14}.
\end{itemize}
}

\short{
\vspace{-2mm}
\begin{itemize}[leftmargin=4ex]
	\item For the CoveringLSH schemes, including fcLSH and bcLSH~\cite{Pagh_SODA16}, we only need the partition trick when $r$ is large (say, $r \geq 10$) since we might not have enough space for $L = 2^{r+1}-1$ hash tables.	
	\vspace{-1mm}
	\item For classic LSH, we simply set the number of hash tables $L = 2^{r+1}-1$ for the sake of comparison. The number of bit samples is set as  $k = \left\lceil \log{(1 - \delta^{1/L})} / \log{(1-r/d)}\right\rceil$ where $\delta$ is the false negative ratio\footnote{http://www.mit.edu/$\sim$andoni/LSH/manual.pdf}.
	\vspace{-1mm}
	\item For MIH, the number of partitions is $\left\lceil d / \log_2{n}\right\rceil$ as suggested in~\cite{Norouzi_PAMI14}.
\end{itemize}
}
\short{\vspace{-2mm}}

\textbf{Cost measurement.} To report all near neighbors, we \full{need to} follow the Strategy~2. In general, for each query, any hashing-based approach needs to process the following operations:
\short{\vspace{-2mm}}
\begin{itemize}[leftmargin=4ex]
	\item \textbf{Step S1:} Compute hash functions to identify the bucket of the query on each of the $L$ hash tables.
	\short{\vspace{-1mm}}
	\item \textbf{Step S2:} Look up in each hash table the points in the bucket of the query, and merge them together for duplicate elimination to form a list of candidates.
	\short{\vspace{-1mm}}
	\item \textbf{Step S3:} Compute the actual distance between candidates and the query to report near neighbor points.
\end{itemize}
\short{\vspace{-2mm}}

We decompose the total search cost per query into~3 cost components of the three main steps above. The cost of~S1 is dependent on the dimensionality of data and the parameter settings for each algorithm which can be analyzed precisely (see Table~\ref{tb:complexity}), whereas the costs of~S2 and~S3 significantly depend on the data distribution and the distance distribution between query and data points, respectively. Since the data sets used in our experiment are both in low-dimensional and high-dimensional space, we focus on the cost of~S2 and~S3.

The cost of~S2, called $C_{lookup}$, is for merging and removing duplicates since very close points might collide many times in different hash tables. Typically, we use a bitmap string of $n$ bits to remove such duplicates~\cite{Norouzi_PAMI14, Sundaram_VLDB13}. Every time a candidate is found, we set the bit corresponding to that candidate. Thus this cost is proportional to the number of collisions \texttt{\#Collisions} over all hash tables.

The cost of~S3, called $C_{check}$, is proportional to the number of \textit{distinct} candidates \texttt{\#Candidates} returned from step~S2. Dependent on the dimensionality $d$, the size of candidates, and cache and disk access implementation, this cost may or may not dominate $C_{lookup}$. Hence, for the sake of comparison, we report separately these two main costs for each algorithm. 

\begin{figure} [t]
\centering
\includegraphics[width=1.0\columnwidth]{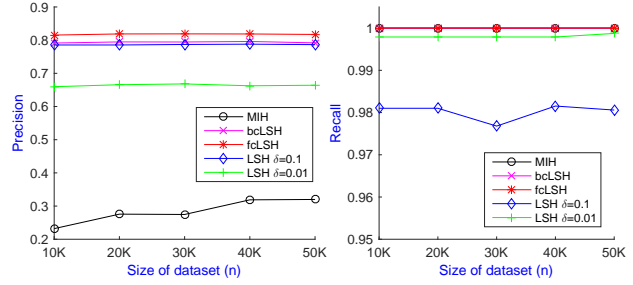}
\caption{Comparison of precision/recall rate between fcLSH and bcLSH without pre-processing, MIH, the classic LSH with $\delta = 0.1$ and $\delta = 0.01$ on synthetic data sets of $n = 10K - 50K$ and $r = 6$.}
\label{fig:SyntheticIdeal}
\end{figure}

\begin{figure*} [ht]
\centering
\includegraphics[width=1.0\textwidth]{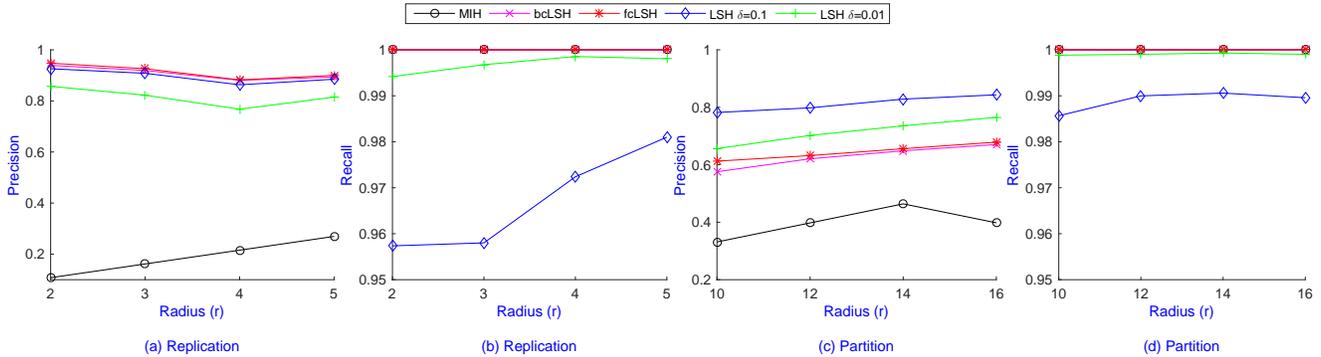}
\caption{Comparison of precision/recall rate between fcLSH and bcLSH with pre-processing, MIH, classic LSH with $\delta = 0.1$ and $\delta = 0.01$ on the synthetic data set of $n = 64K$ points.}
\label{fig:SyntheticPreprocessing}
\end{figure*}

\subsection{Data Sets}

We evaluated the performance of fcLSH using synthetic data sets and~4 real-world data sets from images, text, and recommendation systems. Properties of these data sets are summarized in Table~\ref{tb:datasets}\short{.}\full{, and presented in more detail below.} 
%
\short{\vspace{-2mm}}
\begin{itemize}[leftmargin=4ex]
	\item \textbf{Synthetic} contains uniformly distributed binary data sets of dimension~128. \full{Given the query point, we generated uniformly distributed binary vectors in Hamming balls of radii from~1 to~128. Since the MIH approach runs in sub-linear time for uniformly distributed binary vectors, we use these data sets to verify its performance. In addition, we also compare the performance of the basic $r$-covering scheme to fcLSH with the replication and partitioning trick.}
\short{\vspace{-1mm}}
	\item \textbf{ANN\_SIFT1M~\cite{Jegou_PAMI11}} contains~1 million 128-dimensional SIFT feature vectors of images. \full{We generate standard binary codes~\cite{Charikar_STOC02} to each image to get 64-bit and 128-bit fingerprints of vectors. The fingerprints have the property that if two original images are similar, then the Hamming distance between their fingerprints is small.}
\short{\vspace{-1mm}}
	\item \textbf{Webspam\footnote{http://www.csie.ntu.edu.tw/$\sim$cjlin/libsvmtools/datasets}} contains 350,000 web documents. \full{We apply the standard LSH~\cite{Charikar_STOC02} to each document to get 256-bit and 512-bit fingerprint vectors. The fingerprints have the property that if two original documents are near-duplicates, then the Hamming distance between their fingerprints is small.}
\short{\vspace{-1mm}}
	\item \textbf{Enron\footnote{http://archive.ics.uci.edu/ml/data sets/Bag+of+Words}} contains a collection of about~40,000 emails. \full{After tokenization and removal of stop words, the vocabulary of unique words was binarized by only keeping words that occurred more than ten times. We obtained a very high-dimensional binary text document with $d = 28,102$ unique words.}
\short{\vspace{-4mm}}
	\item \textbf{MovieLens\footnote{http://grouplens.org/data sets/movielens/}} contains ratings applied to 140,214 movies by 234,834 users. \full{Ratings are made on a 5-star scale, with half-star increments (0.5 stars - 5.0 stars). We binarized this data set by only considering ratings at least 2.5 to represent  `positive' and ratings smaller than~2.5 or no ratings for `negative'.  We obtained a very high-dimensional binary data set where each data point corresponds to a user, and we would like to find similar users given their movie ratings.}
\end{itemize}
\short{\vspace{-2mm}}

We randomly remove~50 points from the ANN\_SIFT1M and Webspam data sets, and~100 points from the Enron and MovieLens data sets to use them as query points in our performance study. We need more query points for the latter cases for the sake of comparison, since with small radius (up to~20), there are some query points that do not have any near neighbors. The ground truth for each query point is computed by a linear scan of the entire data sets.

\subsection{Synthetic Data Sets}

We carried out experiments to evaluate the accuracy and efficiency of our constructions with and without a pre-processing step (replicating/partitioning),  over synthetic data sets for the task of reporting all near neighbors. 
We used precision/recall rates to measure the performance of hashing-based methods, including fcLSH, bcLSH, MIH, and classic LSH with recall ratio of 90\% (i.e., $\delta = 0.1$) and 99\% (i.e., $\delta = 0.01$), for a wide range of query radii and data set sizes. We note that if we ignore the $C_{lookup}$ cost, the precision ratio corresponds to the speedup compared to linear search. 
%

%

Figure~\ref{fig:SyntheticIdeal} displays the precision/recall rate of algorithms for reporting points within distance $r = 6$ from a query. The number of hash tables for the LSH-based method is $L = 2^{r+1} - 1 = 127$ whereas that of MIH is at most~10. It is obvious that LSH-based approaches achieve almost~3 times higher precision than the MIH approach. In other words, $C_{check}$ of MIH is around~3 times larger than LSH-based approaches. In addition, fcLSH achieves slightly better precision than both bcLSH and classic LSH. Classic LSH shows a tradeoff between precision and recall rate where the one with recall ratio~99\% has lower precision than that of recall ratio~90\%. Regarding recall ratio, both CoveringLSH schemes and MIH achieve perfect recall whereas classic LSH obtains a high recall ratio (at least~97.5\%) but not~100\%.
%
%
%
%


Figure~\ref{fig:SyntheticPreprocessing} shows the precision/recall rate of fcLSH and bcLSH with preprocessing tricks (replication and partition) and other algorithms. We replicated \{4, 3, 2, 2\} times corresponding to $r = 2,3,4,5$, respectively. This leads to space overhead $L = 511, 1023, 511, 2047$ for LSH-based methods compared to $L = 8$ of MIH, and explains why the precision ratio of MIH fluctuates. The results are very similar to the case without a pre-processing step: LSH schemes show their superiority compared to MIH, fcLSH has slightly higher precision than bcLSH, and the classic LSH approach always introduces false negatives. 
%
%
%
%
We used~2 partitions for $r =10, 12, 14, 16$. The number of hash tables is $L = 126, 254, 510, 1022$ for LSH-based methods, and $L = 8$ for MIH. Again, LSH-based approaches outperform MIH regarding precision ratio. However, the precision of $r$-covering approaches is worse than classic LSH approaches since the partition trick introduces more unexpected collisions. This difference is at most $n^{\ln{(4)}}$ in the worst-case data sets as analyzed in~\cite{Pagh_SODA16}.

Figure~\ref{fig:SyntheticTime} concludes the experimental results on synthetic data sets by showing the hash function computation time per query between two approaches: fcLSH and bcLSH. It is clear that fcLSH gives substantially faster hash function computation time due to the fast Hadamard transform for a wide range of $d$ and $r$. 

\subsection{Real-world Data Sets}

The experiments on synthetic data sets illustrate that fcLSH achieves better performance than bcLSH: less hash function computation time and higher precision with total recall. Hence, we now use fcLSH as the representative of $r$-covering LSH to compare to other approaches on real-world data sets. Since the recall ratios of classic LSH with $\delta = 0.1$ and with $\delta = 0.01$ are almost the same and very high, we only use the classic LSH with $\delta = 0.1$ for comparison. 

We observe that the replication trick often results in more collisions since it uses more hash tables. In practice, the pruning power of LSH-based approaches is primarily dependent on the distance distribution between data points and query points. 
Moreover, the space usage for indexes is usually limited by RAM. This requires the query radius~$r$ to be rather small (say, up to~10) for large data sets (up to~1M points). Therefore, we do not usually need the pre-processing step for small $r$ and only use the partition trick for large $r$. 

As discussed in Subsection~\ref{sec:expsetting}, we used the total number of collisions and the distinct candidate set size, denoted by \texttt{\#Collisions} and \texttt{\#Candidates}, respectively, to measure separately the two main costs $C_{lookup}$ and $C_{check}$. Due to memory constraints we only consider search radius up to~20 on all data sets, except the Enron data set. We use~1 partition (without pre-processing step) for $r < 10$ with $L = 2^{r+1} - 1$, and~2 partitions for $r \geq 10$ with $L = 2(2^{\left\lfloor r/2 \right\rfloor + 1} - 1)$ for LSH-based methods. For MIH, we used the standard setting, i.e., $L = \left\lceil d / \log_2{n}\right\rceil$ hash tables. 

\begin{figure} [t]
\centering
\includegraphics[width=1.0\columnwidth]{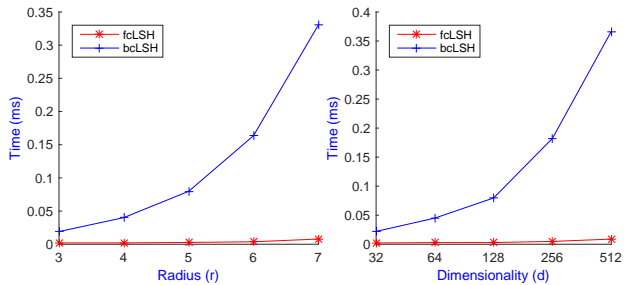}
\caption{Comparison of hash function computation time per query between fcLSH and bcLSH on synthetic data sets of $n = 64K$: $d = 128$ and $r = 3 - 7$ (left), and $r = 5$ and $d = 32 - 512$ (right).}
\label{fig:SyntheticTime}
\end{figure}

\subsubsection{Low-dimensional data sets}

This subsection compares the performance of~4 approaches: covering LSHs (fcLSH and bcLSH), classic LSH with $\delta = 0.1$, and MIH on the ANN\_SIFT1M (images) and Webspam (texts) data sets. Since we aim at measuring the efficiency of these algorithms in low-dimensional space, we generated binary data sets of \{64, 128\} bits for ANN\_SIFT1M, and \{256, 512\} bits for Webspam. 
Due to similar results on both data sets, we only report representative recall ratios of ANN\_SIFT1M~64 bits for small radii $r = 5 - 9$, as shown in Table~\ref{tb:recall_SIFT}. The results confirm that classic LSH cannot avoid false negatives while the other approaches do.
\begin{table}[t]
\centering
\caption{\ Recall ratios on ANN\_SIFT1M 64 bits}
\vspace{2mm}
\label{tb:recall_SIFT} 
\begin{tabular}{|c|c|c|c|c|c|} \hline
Radius & 5 & 6 & 7 & 8 & 9   \\ \hline \hline
fcLSH / MIH & 1 & 1 & 1 & 1 & 1 \\ \hline
Classic LSH & 0.96 & 0.94 & 0.93 & 0.93 & 0.92  \\ \hline
\end{tabular}
\end{table}

\begin{figure*} [!t]
\centering
\includegraphics[width=1.\textwidth]{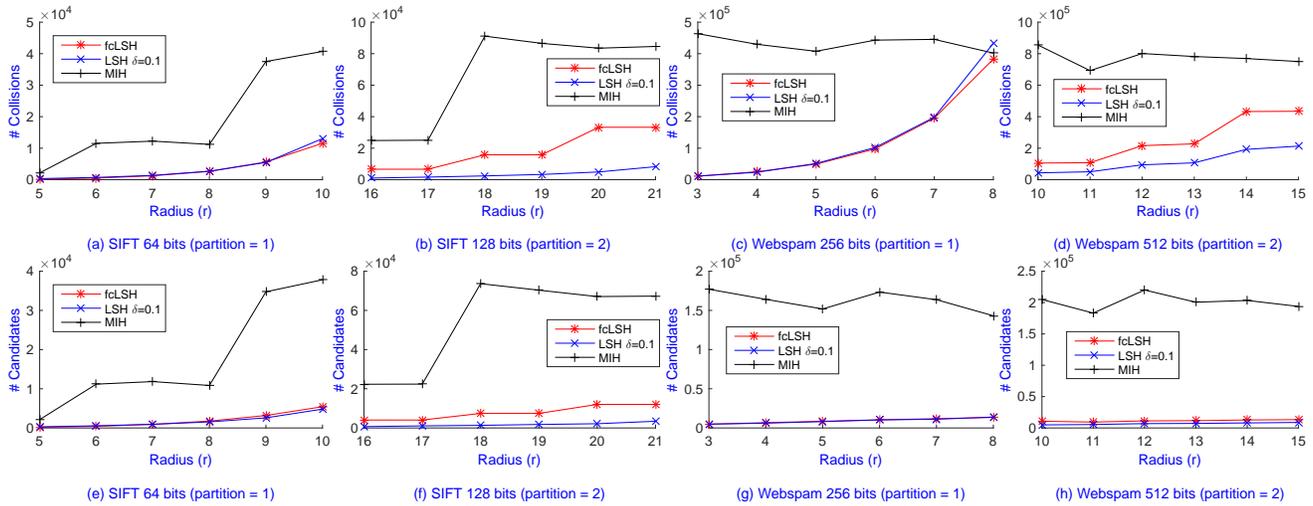}
\caption{Comparison of the number of collisions and distinct candidate set size for fcLSH, MIH, classic LSH with $\delta = 0.1$ on two data sets: ANN\_SIFT1M and Webspam.}
\label{fig:SIFT_Webspam_Cost}
\end{figure*}

\begin{figure*} [t]
\centering
\includegraphics[width=1.0\textwidth]{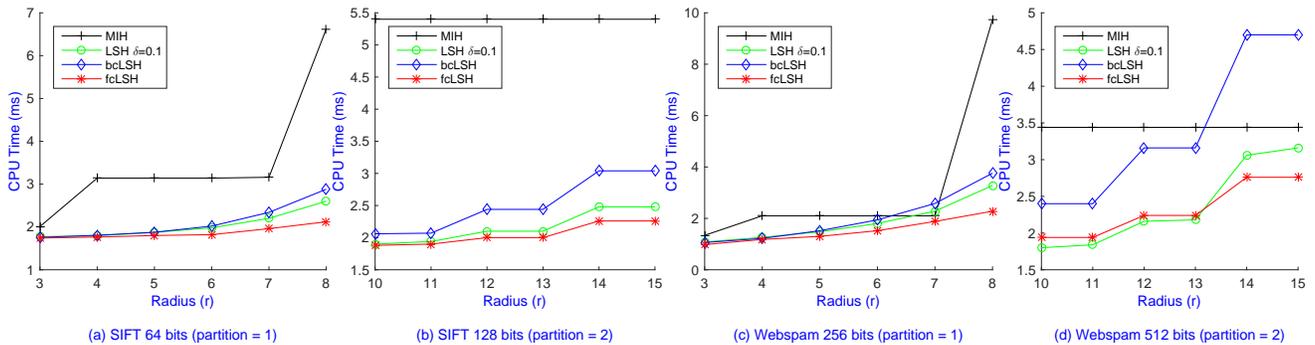}
\caption{Comparison of CPU Time (ms) per query of~4 approaches: fcLSH, bcLSH, classic LSH with $\delta = 0.1$, and MIH on two data sets: ANN\_SIFT1M and Webspam.}
\label{fig:SIFT_Webspam_CPU}
\end{figure*}

Figure~\ref{fig:SIFT_Webspam_Cost} shows the two main costs per query on the ANN\_SIFT1M and Webspam data sets with different dimensionality. Since fcLSH and bcLSH have the same hash values, the two main costs are identical. It is obvious that LSH-based approaches outperform the MIH approach on the ANN\_SIFT1M data set. For the 64-bit version, \texttt{\#Candidates} and \texttt{\#Collisions} for MIH are dramatically larger than for fcLSH and classic LSH. In particular, MIH's \texttt{\#Collisions} is up to around~7 times larger than that of LSH-based approaches. The largest gap starts at $r = 9$. This is because MIH uses~3 hash tables and $r \geq 9$ indicates a new radius $r' = 3$ for each partition. This change influences both \texttt{\#Candidates} and \texttt{\#Collisions} of MIH. 

As the theoretical analysis shows that $r$-covering LSH schemes and classic LSH have  similar pruning power for far points, their performance, including \texttt{\#Candidates} and \texttt{\#Collisions}, are very similar for $r = 5 - 10$. For~2 partitions, CoveringLSH is slightly worse than classic LSH due to the probability of splitting distances unevenly over the partitions. However, when we target to approach~100\% recall ratio, covering LSH schemes clearly outperform MIH, obtaining up to~7 and~14 times speedup regarding \texttt{\#Collisions} and \texttt{\#Candidates}, respectively.

On the Webspam data set, \texttt{\#Candidates} returned by MIH is orders of magnitude larger than for the LSH-based approaches. This is because $d$ is rather large, so the estimated cost of MIH, $\BO{(d/r)^r}$, tends to be very large, even comparable to the data set size. Hence, in terms of guaranteeing perfect recall, CoveringLSH provides superior performance compared to MIH. Compared to classic LSH, the performance of CoveringLSH is similar when using~1 partition and slightly worse with~2 partitions. In particular, \texttt{\#Candidates} and \texttt{\#Collisions} provided by fcLSH using~2 partitions is approximately twice that of classic LSH.

Figure~\ref{fig:SIFT_Webspam_CPU} shows the superiority of LSH-based methods (fcLSH, bcLSH, and classic LSH with $\delta = 0.1$) to the MIH method with respect to the average CPU time per query in milliseconds on the ANN\_SIFT1M and Webspam data sets. We note that the standard setting of MIH requiring number of hash tables $L = \left\lceil d / \log_2{n}\right\rceil$ does not result in a good performance since the real-world data sets are not uniformly distributed. For the sake of comparison, we choose $L = 4, 8$ corresponding to the two different versions of these data sets, that leads to the best performance of MIH. Even in such least favorable scenario, LSH-based approaches still run at least~2 times faster than MIH on the ANN\_SIFT1M data sets. On the~64-bit version, since the cost $C_{check}$ and $C_{lookup}$ of LSH-based approaches are very similar, fcLSH provides superior performance compared to bcLSH and classic LSH due to the fast hash computation. For $r = 6 - 8$, bcLSH is slightly slower than classic LSH. This is because the increase in the of number of hash tables, $L = 2^{r+1} - 1$ leads to a slightly larger gap in hash computation time, $dL$ of bcLSH compared to $kL$ of classic LSH. On the~128-bit version, classic LSH is favorably compared to bcLSH because \texttt{\#Candidates} and \texttt{\#Collisions} provided by bcLSH considerably increase due to partitioning. However, fcLSH still gains substantial advantages from the fast hash computation and outperforms bcLSH and classic LSH.

On the Webspam~256-bit dataset, MIH is slightly slower than LSH-based approaches for small radii $r = 3 - 7$. This CPU time gap is more significant at the radius $r = 8$ because this new radius yields to a new radius $r' = 2$ on each partition of MIH, noting that MIH uses $L = 4$. This degrades the performance of MIH due to the significant growth of \texttt{\#Candidates} and \texttt{\#Collisions}. It is worth noting that this observation is also illustrated in Figure~\ref{fig:SIFT_Webspam_Cost} when MIH uses the standard setting $L = \left\lceil d / \log_2{n}\right\rceil$. On the Webspam~512-bit version, both fcLSH and classic LSH outperform MIH for $r = 10 - 15$. Moreover, fcLSH is comparable to classic LSH for $r = 10 - 13$, but is superior for $r = 14 - 15$ since the hash computation time dramatically contributes to the total cost. That also explains why bcLSH is worse than MIH on this parameter setting. In general, fcLSH is favorable compared to the other approaches regarding both CPU time and total recall.


%
\begin{figure*} [t]
\centering
\includegraphics[width=1.0\textwidth]{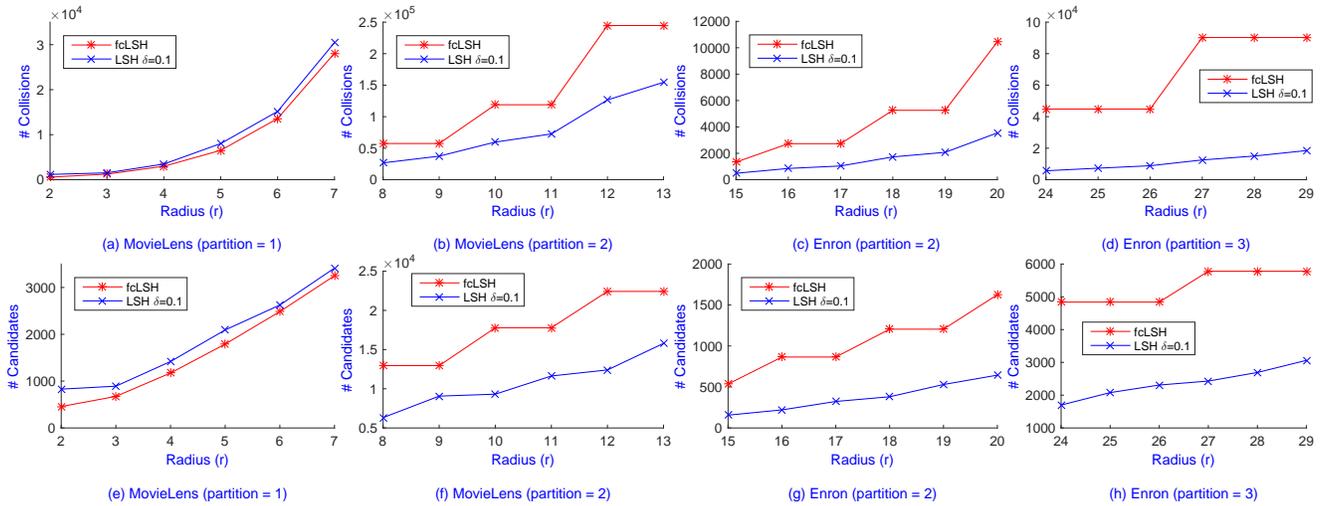}
\caption{Comparison of the number of collisions and distinct candidate set size between fcLSH and classic LSH with $\delta = 0.1$ on two data sets: MovieLens and Enron.}
\label{fig:Enron_MovieLens_Cost}
\end{figure*}

\begin{figure*} [t]
\centering
\includegraphics[width=1.0\textwidth]{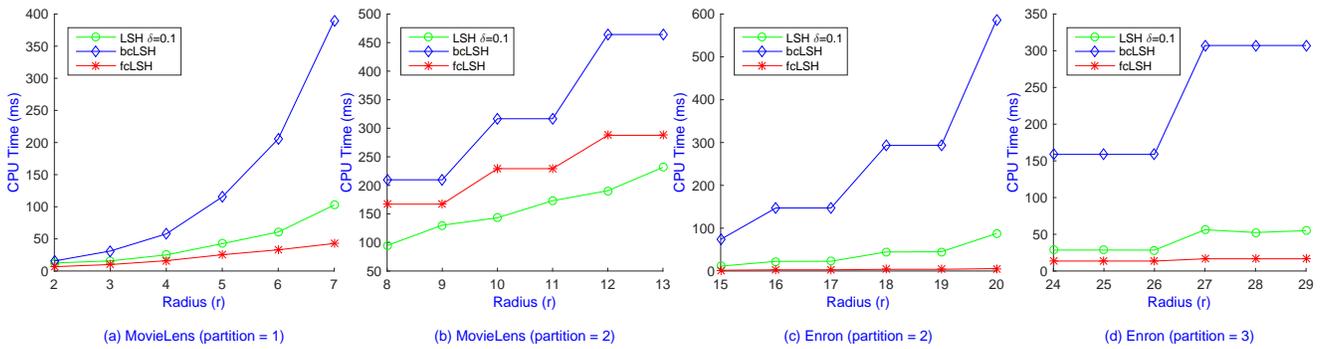}
\caption{Comparison of CPU Time (ms) per query of~3 approaches: fcLSH, bcLSH, and classic LSH with $\delta = 0.1$ on two data sets: Movielens and Enron.}
\label{fig:Enron_MovieLens_CPU}
\end{figure*}

\subsubsection{High-dimensional data sets}

This subsection studies the ability of scale and accuracy of~3 approaches, fcLSH, bcLSH and classic LSH with $\delta = 0.1$ on the two high-dimensional binary data sets: MovieLens and Enron. Since the data sets are very high-dimensional, the MIH approach is outperformed by the simple linear search and we do not report the results for MIH here. Due to similar results on the two data sets, we report representative recall ratios of MovieLens for small radii $r = 3 - 7$, as shown in Table~\ref{tb:recall_movielens}. The results once again confirm that fcLSH always eliminates false negatives while classic LSH cannot.
\begin{table}[t]
\centering
\caption{\ Recall ratios on Movielens}
\label{tb:recall_movielens} 
\vspace{2mm}
\begin{tabular}{|c|c|c|c|c|c|} \hline
Radius & 3 & 4 & 5 & 6 & 7   \\ \hline \hline
fcLSH & 1 & 1 & 1 & 1 & 1 \\ \hline
Classic LSH & 0.97 & 0.99 & 0.99 & 0.98 &  0.98  \\ \hline
\end{tabular}
\vspace{-1mm}
\end{table}

For the MovieLens data set, we use~1 and~2 partitions for $r = 2 - 7$ and $r = 8 - 13$, respectively. Since the Enron data set is rather small, we can use~3 partitions and measure the performance of fcLSH with radius up to~29. It is worth noting that the data sets are very high-dimensional and distance computation is time consuming, the cost $C_{check}$ dominates the cost $C_{lookup}$. Hence, we focused on discussing \texttt{\#Candidates} and CPU Time (ms) per query of the~3 approaches. We again used \texttt{\#Collisions} and \texttt{\#Candidates} to measure the costs $C_{lookup}$ and $C_{check}$, respectively, of the~3 approaches, as shown in Figure~\ref{fig:Enron_MovieLens_Cost}. This result again supports our theoretical comparison of fcLSH to classic LSH. \texttt{\#Candidates} of CoveringLSH is slightly smaller than classic LSH in the case of~1 partition on MovieLens but up to~3 times larger when using~2 and~3 partitions on Enron. 

Figure~\ref{fig:Enron_MovieLens_CPU} shows superiority of fcLSH to bcLSH with respect to the CPU time in milliseconds due to the fast hash computation time on the two data sets. On the MovieLens data set, fcLSH is faster than classic LSH with~1 partition but is slower with~2 partitions. This is because \texttt{\#Candidates} provided by fcLSH using~2 partitions is approximately twice larger than that of classic LSH. However, on the Enron data set, fcLSH outperforms classic LSH even though it uses partitioning trick. This is due to the fact that \texttt{\#Candidates} on Enron is rather small and the hash computation time dominates the total running time. In conclusion, fcLSH is favorably compared to classic LSH but is superior to both MIH and bcLSH in settings requiring precise performance guarantees.

\full{
\section{Related Work}\label{sec:relatedwork}

Due to the ``curse of dimensionality'', one typically uses linear search for (exact) near neighbor search in high-dimensional Hamming space~\cite{Norouzi_NIPS12, Salakhutdinov_IJAR09, Torralba_CVPR08}. To trade precision for speed, \textit{approximate} retrieval is widely investigated in the research literature, and LSH~\cite{Indyk_STOC98} is a widely used technique due to its attractive ``tradeoff'' between time and space. However, false negatives findings limit the applicability of LSH in settings requiring precise performance guarantees.


Recently, Norouzi et. al.~\cite{Norouzi_PAMI14} proposed the MIH approach which partitions each data vector to reduce the search radius, and then applies exhaustive search. Although the MIH approach has sub-linear running time behaviour for \emph{uniformly} distributed data sets, it does not work well in general. This is because its performance relies on the ability to select a small number of random bit positions (around $\log n$) for which there are almost no collisions between the query point and points in the data set -- an assumption that is not true in general. An approach similar to MIH was taken in~\cite{Liu_ICDE11}, with the same vulnerability. We have chosen to compare against MIH as a representative of these schemes.

Arasu et al.~\cite{Arasu_VLDB06} proposed the idea that randomly permuting the dimensions of data vectors increases the robustness of partitioning, and make performance guarantees possible for data sets that are not uniformly distributed. They combined this idea with another level of partitioning within which a ``brute force'' $r$-cover is found. For $r > 2$ the scheme is never better than CoveringLSH~\cite{Pagh_SODA16}. In the case of a single partition the number of hash values needed by~\cite{Arasu_VLDB06} is $\binom{2r}{r}\approx 4^r/\sqrt{r}$, which is much larger than $2^{r+1}$ required by our Hadamard code-based method for the same filtering efficiency. In the case of several partitions, Arasu et al. leave it unspecified how to best choose the parameters of their method, so it is really a family of methods. For these reasons we have not implemented this method.
}

\section{Conclusions}\label{sec:conclusion}

This paper proposes \textit{Fast CoveringLSH}, a fast and practical LSH scheme for Hamming space. Inheriting the design benefits from CoveringLSH, our method avoids false negatives and always reports \textit{all} near neighbors. Our main technical contribution is asymptotic improvement to the hash function computation time from $\BO{dL}$ to $\BO{d + L \log{L}}$, for $d$ dimensions and $L$ hash tables. 
Our experiments on synthetic and real-world data sets demonstrate the efficiency of fcLSH in comparison with traditional hashing-based approaches for search radius up to~20 in high-dimensional Hamming space. 

An obvious open direction is to extend our work to other spaces or similarity measures, aiming at rigorous performance guarantees without false negatives. 
Since the recent covering LSH framework demands a large number of hash tables for large radii, another interesting question would be to reduce the space usage to linear (or near-linear) in the data size while maintaining the property of total recall.




%
\bibliographystyle{abbrv}
\bibliography{sigproc}  

%
%

\end{document}